\documentclass{uai2023} 

\usepackage[american]{babel}

\usepackage{natbib} 
\usepackage{mathtools} 
\usepackage{booktabs} 
\usepackage{tikz} 

\usepackage{amsthm}
\usepackage{multirow}
\usepackage{subcaption}
\newtheorem{definition}{Definition}
\newtheorem{theorem}{Theorem}
\newtheorem{lemma}{Lemma}
\newtheorem{proposition}{Proposition}
\newtheorem{corollary}{Corollary}
\newtheorem{assumption}{Assumption}

\title{A Truthful Referral Auction Over Networks}

\author[]{Youjia Zhang}
\author[]{Pingzhong Tang}
\affil[]{
    Institute for Interdisciplinary Information Sciences\\ 
    
    Tsinghua University\\
    
    zhangyou19@mails.tsinghua.edu.cn, kenshinping@gmail.com 
}

\begin{document}
\maketitle

\begin{abstract}
  This paper studies a mechanism design problem over a network, where agents can only participate by referrals. The Bulow-Klemberer theorem proposes that expanding the number of participants is a more effective approach to increase revenue than modifying the auction format. However, agents lack the motivation to invite others because doing so intensifies competition among them. On the other hand, misreporting social networks is also a common problem that can reduce revenue. Examples of misreporting include Sybil attacks (an agent pretending to be multiple bidders) and coalition groups (multiple agents pretending to be an agent). To address these challenges, we introduce a novel mechanism called the Truthful Referral Diffusion Mechanism (TRDM). TRDM incentivizes agents to report their social networks truthfully, and some of them are rewarded by the seller for improving revenue. In spite of the fact that some agents overbid in TRDM, the revenue is fixed, and it is higher than the revenue of any mechanism without referrals. TRDM is budget-balanced (non-negative revenue) and generates an efficient outcome (maximized social welfare), making it attractive for both the seller and the buyers as it improves revenue and reward.
\end{abstract}

\section{Introduction}
\label{sec: intro}

Mechanism design is an interdisciplinary field that integrates insights from economics, game theory, and computer science to design mechanisms that can achieve desirable outcomes in a strategic environment where all agents act rationally. Auction is a popular application of mechanism design that determines the fair market value of goods and services. Different formats of auctions are designed to achieve specific objectives such as individual rationality, incentive compatibility, revenue maximization, and efficiency. Auction design is a popular area of study that draws a lot of attention and research from economists and other experts. The 2020 Nobel Prize for Economics was awarded to Paul R. Milgrom and Robert B. Wilson “for improvements to auction theory and inventions of new auction formats.” \cite{nobel2020}

In traditional mechanism design, the set of agents is predetermined and assumed to be independent. However, with the advancement of technology, the prevalence and impact of social media have significantly changed the way people communicate and access information. Nowadays, individuals interact with one another more frequently and efficiently than ever before. As a result, many companies are reevaluating their marketing strategies and shifting their focus toward utilizing social networks. Instead of traditional advertising, they offer incentives for customers to refer their friends. As an example, Dropbox has increased its total user base to more than 4 million in less than two years by implementing a referral program \cite{dropbox}. 

Hence, it is natural to develop auctions over social networks as it benefits all agents, including the seller. Particularly, the Bulow-Klemperer theorem argues a Vickrey auction with an additional bidder can generate more revenue than Myerson's optimal auction \cite{bulow1996auctions}. In real-world scenarios, sellers on eBay pay to extend the auction duration and advertise the sale information on the site in order to attract more participants. However, these actions may not bring valuable buyers and increase revenue, making these investments ineffective. 

Instead of designing complicated auction formats, we propose a mechanism over social networks, where agents join the auction through referrals. In traditional auctions such as the First Price auction (FPA) and the Second Price auction (SPA), agents lack incentives to invite others, as new agents may lead existing agents to lose the item or increase the payment. Therefore, the mechanism should reward agents who contribute to the growth of revenue. Recent literature has referred to mechanism design over social networks as \textit{diffusion mechanism} and has demonstrated it to be a promising area of research \cite{li2017mechanism,li2019diffusion,zhang2020incentivize,zhang2020redistribution,li2021incentive}.

Furthermore, agents may benefit from misreporting their social networks. For example, an agent creates multiple fake identities (Sybil attacks) to pay less for the item or receive more rewards from the seller. On the other hand, in network settings, it is easier for agents to make a deal with others and pretend to be a single agent (collusion group) in order to improve their utility. Both Sybil attacks and collusion group are well-defined in \cite{yokoo2001robust, nath2012mechanism}, respectively, but have not been extensively studied in the field of diffusion mechanism.

This paper proposes a novel mechanism called the Truthful Referral Diffusion Mechanism (TRDM) to overcome the aforementioned challenges. Our mechanism not only ensures all agents report their networks truthfully but also guarantees an efficient allocation and generates higher revenue than auctions without referrals. 

The remainder of the paper is organized as follows. Section \ref{sec: literature review} reviews the relevant literature. Section \ref{sec: preliminaries and model} outlines the basic concepts and background information of the model, including the impossibility results. In Section \ref{sec: mechanism}, we propose our novel mechanism in detail and show its properties. In Section \ref{sec: comparison}, a numerical comparison of our mechanism with other diffusion mechanisms is presented. Finally, we provide concluding comments and discuss potential future research in Section \ref{sec: conclusion}.

\section{Literature Review}
\label{sec: literature review}

There is an increasing amount of literature on mechanism design in network settings. For more information, we recommend checking out the recent work by \cite{zhao2021mechanism} and the references cited within it.

The seminal paper \cite{li2017mechanism} established that extending the VCG mechanism \cite{vickrey1961counterspeculation,clarke1971multipart,groves1973incentives} to social networks results in budget deficits. They were the first to propose an auction on social networks, which is called the information diffusion mechanism (IDM). IDM incentivizes participants to invite others and ensures non-negative revenue. Later, \cite{zhao2018selling} extended the mechanism for selling multiple items on a social network. Following the seminal work of \cite{li2017mechanism}, \cite{li2019diffusion} further identified a class of diffusion mechanisms under social networks named Critical Diffusion Mechanisms (CDM), of which IDM is a specific case that results in the lowest revenue. Recently, \cite{li2021incentive} identified a condition for all dominant-strategy incentive-compatible (DSIC) diffusion auctions by decoupling the payment function.

Inspired by the redistribution mechanism in \cite{cavallo2006optimal}, \cite{zhang2020incentivize} proposed the Fair Diffusion Mechanism (FDM) that compensates more participants fairly without reducing the seller's revenue. Subsequently, \cite{zhang2020redistribution} extended the FDM to a more general setting, resulting in the Network-based Redistribution Mechanism(NRM) that satisfies both IC and IR constraints and is budget feasible.

The works mentioned above have not been explored with regards to Sybil attacks and collusion groups. Other research fields such as \cite{emek2011mechanisms, babaioff2012bitcoin, chen2013sybil, nath2012mechanism, zhang2021sybil, zhang2023collusionproof} have demonstrated the potential of Sybil-proof and collusion-proof mechanisms, making them promising areas of research for further investigation.

\section{Preliminaries and Model}
\label{sec: preliminaries and model}

Considering an auction with a seller $s$ and $n$ agents. Let $N$ be the set of all agents, including the seller. 

For each agent $i \in N$, he has a type $\theta_{i}=(r_{i}, v_{i})$, where $r_{i} \subseteq \{ N \setminus s\}$ denotes the set of children (friends) of agent $i$ (whom $i$ can directly invite), and $v_{i}$ is agent $i$'s valuation for the item. Let $\theta=(\theta_{1},...,\theta_{n})=(\theta_{i}, \theta_{-i})$ be the type profile of all agents, where $\theta_{-i}$ is the type profile of all agents excluding agent $i$.

Agents are asked to report their types to the seller. We denote $\theta_{i}^{\prime}=(r_{i}^{\prime}, v_{i}^{\prime})$ as the report type of agent $i$, where $v_{i}^{\prime} \in R^{+}$ is non-negative. In particular, it is impossible to share the information of the mechanism to a non-existing child, so $r_{i}^{\prime} \subseteq r_{i}$. Similarly, $\theta^{\prime}=(\theta_{1}^{\prime}, ..., \theta_{n}^{\prime})=(\theta_{i}^{\prime}, \theta_{-i}^{\prime})$ is the reported profile of all agents. 

Given the reported profile $\theta^{\prime}$, we can generate a directed graph $G(\theta^{\prime})=(V(\theta^{\prime}), E(\theta^{\prime}))$, where $V(\theta^{\prime}) \subseteq N$ is the set of agents who join the auction and an edge $e(i, j) \in E(\theta^{\prime})$ represents the connection between agents $i$ and $j$ (e.g., $j \in r_{i}$). Agent $i$ is \textit{connected} if there exists a path $s \rightarrow ... \rightarrow i$ in the directed graph $G(\theta^{\prime})$.

\begin{definition}
\label{def: mechanism}
    A diffusion mechanism $M$ is a pair of allocation rule and payment rule, $M=\langle x,p \rangle$, an allocation rule $x=(x_{1}, ..., x_{n})$ and a payment rule $p=(p_{1},...,p_{n})$, where $x_{i}: \Theta \rightarrow \{0,1\}$ and $p_{i}: \Theta \rightarrow \mathbf{R}$.
\end{definition}   

Given the reported profiles $\theta^{\prime} \in \Theta$, agent $i$ is allocated the item if $x_{i}(\theta^{\prime})=1$ (we refer such an agent as the \textbf{winner}), otherwise $x_{i}(\theta^{\prime})=0$. In the meanwhile, $p_{i}(\theta^{\prime}) > 0$ indicates that agent $i$ pays to the seller, and if $p_{i}(\theta^{\prime})<0$, agent $i$ receives a \textbf{reward} from the seller. Therefore, the utility of an agent with $\theta_{i}$ reporting $\theta_{i}^{\prime}$ can be written as

\begin{equation}
\label{eq: utility}
    \begin{aligned}
        u_{i}(\theta^{\prime})=x_{i}(\theta^{\prime}) \cdot v_{i} - p_{i}(\theta^{\prime}).
    \end{aligned}
\end{equation}

In traditional mechanism design, two fundamental properties are considered: \textbf{individual rationality (IR)} and \textbf{incentive compatibility (IC)}. The former ensures that the mechanism is rational for all agents to participate in, while the latter ensures that agents have incentives to report their types truthfully. In social network settings, the type space is expanded to include the act of inviting. Hence, we also extend the definitions of IC and IR to include it.

\begin{definition}
\label{def: IR}

    The diffusion mechanism $M=\langle x,p \rangle$ is \textbf{individually rational (IR)} if for all $i \in N$, it holds that $u_{i}(\theta) \geq 0$.

\end{definition}

\begin{definition}
\label{def: IC}
    The diffusion mechanism $M=\langle x,p \rangle$ is \textbf{incentive compatible (IC)} if $u_{i}(\theta_{i}, \theta_{-i})\geq u_{i}(\theta_{i}^{\prime}, \theta_{-i}^{*})$, for all $i \in N$ and for all $j \in V(\theta_{i}^{\prime}, \theta_{-i}^{*})$ and $j \neq i$, $\theta_{j}^{*}=\theta_{j}$.
    
    Moreover, mechanism $M$ is \textbf{partially incentive compatible} if satisfying one of the following conditions
    \begin{itemize}
        \item (\textbf{truthful bid}) $u_{i}((v_{i}, r_{i}^{\prime}), \theta_{-i}^{*})  \geq u_{i}((v_{i}^{\prime}, r_{i}^{\prime}), \theta_{-i}^{*})$.
        \item (\textbf{truthful referral}) $u_{i}((v_{i}^{\prime}, r_{i}), \theta_{-i}) \geq u_{i}((v_{i}^{\prime}, r_{i}^{\prime}), \theta_{-i}^{*})$.
    \end{itemize}
\end{definition}

Note that when agent $i$ misreports $r_{i}^{\prime}$, some of his children and descendants cannot join the auction (not in the directed graph). Hence, we denote $\theta_{-i}^{*}$ as the \textbf{truthful} reported profile of all agents $j \in \{V(\theta_{i}^{\prime}, \theta_{-i}^{*}) \setminus i\}$.

As we mentioned, Sybil attacks in diffusion mechanisms have not been well-studied. Agents may manipulate multiple fake identities to influence the outcome of a mechanism. The following is a formal definition of Sybil-proof diffusion mechanism:

\begin{definition}
\label{def: sybil}
    The diffusion mechanism $M=\langle x,p \rangle$ is \textbf{Sybil-proof (SP)} if for all agents $i \in N$, 
    \begin{equation*}
        \begin{aligned}
            u_{i}((v_{i}^{\prime}, r_{i}^{\prime}), \theta_{-i}) \geq \sum_{j \in F_{i}} u_{j}((v_{j}^{*}, r_{j}^{*}), \theta_{-F_{i}}),
        \end{aligned}
    \end{equation*}
    where $v_{i}^{\prime}, v_{i}^{*} \in \mathbf{R}_{+}$, $r_{i}^{\prime} \subseteq r_{i}$, $r_{i}^{*} \subseteq r_{i} \cup F_{i}$, and $F_{i}$ are fake accounts owned by agent $i$ including himself. 
\end{definition}

In network settings, it is easier for agents to communicate with their parents and children to form a coalition group (multiple agents pretending to be a single agent) in order to improve their utility.

\begin{definition}
\label{def: collusion}
    The diffusion mechanism $M=\langle x,p \rangle$ is \textbf{collusion-proof (CP)} if for all agents $i \in N$, there is \textbf{no} collusion group $C_{i}$, such that
    \begin{equation*}
        \begin{aligned}
            \sum_{j \in C_{i}} u_{j}((v_{j}^{*}, r_{j}^{*}), \theta_{-C_{i}}) \geq \sum_{j \in C_{i}} u_{j}((v_{j}^{\prime}, r_{j}^{\prime}), \theta_{-j}),
        \end{aligned}
    \end{equation*}
    where $v_{i}^{\prime}, v_{i}^{*} \in \mathbf{R}_{+}$, $r_{i}^{\prime} \subseteq r_{i}$, $r_{i}^{*} \subseteq r_{i}^{\prime}$, and $C_{i}$ is the collusion group rooted at agent $i$. 
\end{definition}

It is important to note that under a Sybil-proof (or collusion-proof) mechanism, agents may also benefit from propagating the auction to some of their children rather than all (non-truthful referrals).

\begin{corollary}
\label{coro: sybil-collusion-truthful}
    A truthful referral mechanism is Sybil-proof and collusion-proof; however, a Sybil-proof (or collusion-proof) mechanism may not be referral truthfulness.
\end{corollary}

For the diffusion mechanism $M$, given the reported profile $\theta^{\prime}$, the overall revenue is equivalent to the sum of all agents' payments and denoted as $Rev^{M}(\theta^{\prime})=\sum_{i \in N}p_{i}(\theta^{\prime})$. 

\begin{definition}
\label{def: budget-balance}
    The diffusion mechanism $M=\langle x,p \rangle$ is \textbf{budget balanced (BB)} if $Rev^{M}(\theta^{\prime})\geq 0$ for all $\theta^{\prime} \in \Theta$, where $Rev^{M}(\theta^{\prime})=\sum_{i \in N}p_{i}(\theta^{\prime})$.
\end{definition}

In addition, the social welfare of an allocation is the sum of the reported valuations of all agents who win the item, and it is defined as $SW^{M}(\theta^{\prime})=\sum_{i \in N}x_{i}(\theta^{\prime})v_{i}$. 

\begin{definition}
\label{def: efficient}
    The diffusion mechanism $M=\langle x,p \rangle$ is \textbf{efficient} if $SW^{M}(\theta^{\prime})=\mathop{\max}\limits_{i \in N}v_{i}$ for all $\theta^{\prime} \in \Theta$, where $SW^{M}(\theta^{\prime})=\sum_{i \in N}x_{i}(\theta^{\prime})v_{i}$.
\end{definition}

In addition, in order to loosen the restriction of the diffusion mechanism, we make the following assumptions.

\begin{assumption}
\label{assumption: kind-item_seeker}
    All agents are kind and item-seekers. Mathematically, they can be written as
    \begin{itemize}
        \item (kindness) if $u_{i}((v_{i}^{\prime}, r_{i}), \theta_{-i}) = u_{i}((v_{i}^{\prime}, r_{i}^{\prime}), \theta_{-i}^{*}) =0$, then $(v_{i}^{\prime}, r_{i}^{\prime})=(v_{i}^{\prime}, r_{i})$
        \item (item-seeker) if $u_{i}((v_{i}, r_{i}^{\prime}), \theta_{-i}^{*})  = u_{i}((v_{i}^{\prime}, r_{i}^{\prime}), \theta_{-i}^{*})=0$, then $(v_{i^{\prime}}, r_{i}^{\prime})=(v_{i}, r_{i})$.
    \end{itemize}
\end{assumption}

The assumption of kindness ensures all agents prefer to spread the sale information even if they have no gain. The assumption of the item-seeker promises all agents would prefer to purchase the item rather than lose it. Intuitively, agents only misreport if it can improve their utilities; otherwise, they report truthfully.

\begin{assumption}
\label{assumption: complete-information}
    Assume all agents know the highest valuation among their descendants.
\end{assumption}

Assumption \ref{assumption: complete-information} simplifies our analysis, and it is a reasonable assumption in real-world scenarios. Without it, agents may be unwilling to invite others as it could negatively impact their utility.

\subsection{Impossibility Results}
\label{sec: impossibility}

Before proposing the novel diffusion mechanism that satisfies the desirable properties, we introduce several impossibility results in this section. The construction of our diffusion mechanism is motivated by the following impossibility theorems.

\begin{theorem}
\label{thm: impossibility}
    It is impossible for any diffusion mechanism to simultaneously achieve efficiency, individual rationality, and incentive compatibility.
\end{theorem}

\begin{proof}
\begin{figure}[htp]
    \centering
    \includegraphics[width=0.2\textwidth]{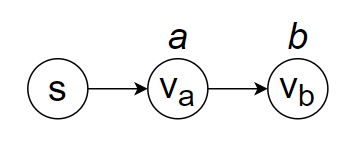}
    \caption{Network example to proof Theorem \ref{thm: impossibility}.}
    \label{fig: simple}
\end{figure}
    The proof is the same as in \cite{zhang2020redistribution}. Consider an example of Figure \ref{fig: simple}. If agent $a$ does not invite agent $b$, $a$ purchases the item with a price between $[0, v_{a}]$. 
    
    If $a$ invites $b$ and $v_{b}>v_{a}$, for an efficient allocation, the item should be allocated to $b$. Moreover, $a$ should be rewarded at least $v_{a}$; otherwise, it is worthless for $a$ to invite $b$. 
    
    As the reward of agent $a$ is related to his bid, he can bid $v_{a}^{\prime}>v_{a}$ to maximize his utility, which contradicts IC.

    On the other hand, if agent $b$ pays more than $v_{b}$, the mechanism fails to achieve IR. However, if the payment of agent $b$ is related to his bid, the mechanism also fails to achieve IC.
    
\end{proof}

Theorem \ref{thm: impossibility} reveals that in social network settings, it is impossible for a mechanism to satisfy all the properties we mentioned in Section \ref{sec: preliminaries and model}. In particular, most of the existing works \cite{li2017mechanism,li2019diffusion,zhang2020incentivize} sacrifice IC and the efficiency of the diffusion mechanism to achieve BB and IR.

\begin{theorem}
\label{thm: upper bound}
    For any diffusion mechanism, the revenue is upper-bounded by the maximum valuation in the sub-network, excluding the group rooted at the seller's child, who is also the ancestor of the winner. Otherwise, the mechanism is \textbf{not} truthful referral.
\end{theorem}

\begin{proof}
    \begin{figure}[htp!]
        \centering
        \includegraphics[width=0.7\linewidth]{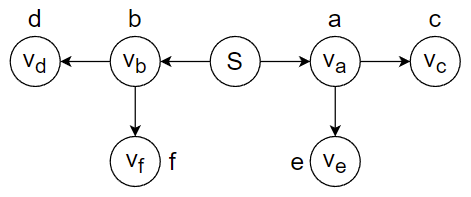}
        \caption{Network example, with valuation relationships $v_{d}>v_{c}>v_{e}>v_{f}>v_{a}>v_{b}$.}
        \label{fig: upper_bound}
    \end{figure}
    Agent $d$ has the highest valuation in the complete network, and agent $c$ has the highest valuation in the sub-network rotted at agent $a$.

    If the seller receives more than $v_{c}$, agent $b$ can collude with his descendants (agents $d$ and $f$) to form a group $C_{b}$. Note that an agent can communicate with his children in the settings of networks. Consider agents in $C_{b}$ pretending to be an agent $b^{*}$ and bidding $v_{b^{*}}^{\prime} = max(v_{b}, v_{d}, v_{f})$. 
    
    Then, agent $b^{*}$ wins the item and pays $v_{c}$, then agents in the group $C_{b}$ redistribute the surplus $(v_{d}-v_{c})$ among themselves.
        
    As a result, the utilities are always higher than those without collusion.
\end{proof}

For any diffusion mechanism in which revenue exceeds the upper bound we defined in Theorem \ref{thm: upper bound}, agents may apply Sybil attacks or form a coalition group such that these agents can improve their utilities without worsening other agents, excluding the seller.

\section{Truthful Referral Diffusion Mechanism}
\label{sec: mechanism}

In this section, to overcome the challenges in mechanism design over social networks, we present the Truthful Referral Diffusion Mechanism (TRDM). TRDM is different from existing mechanisms, which sacrifice efficiency and truthful referral to ensure BB and IR. It outputs an efficient outcome and ensures all agents refer truthfully. Before introducing TRDM in detail, we establish some fundamental definitions and notations:

\begin{definition}
\label{def: efficient diffusion sequence}
    Given the reported profiles $\theta^{\prime}$ and the corresponding directed network $G(\theta^{\prime})$, for each $i \in V(\theta^{\prime})$ (denote $V(\theta^{\prime})$ as $V$), we define $\Gamma_{i}=\{s, a_{1}, a_{2},..., i \}$ as the shortest path from seller $s$ to agent $i$.
\end{definition}

\begin{itemize}
    \item The winner of the item is denoted as $w$. 

    \item For a node $i$, $r^{*}_{i}$ represents a child of node $i$ that is located in the shortest path to the winner $w$, $\Gamma{w}$ (e.g., $r^{*}_{j}=l$ in Figure \ref{subfig: example}).

    \item The set $D_{i}^{V}$ denotes the reachable descendants of node $i$ within the set of agents $V$ (e.g. $D_{j}^{V}=\{l, m, n\}$ in Figure \ref{subfig: social network}).

    \item $v_{V}^{*}$ represents the highest valuation reported among the set of agents $V$ (e.g., $v_{V}^{*}= \mathop{\max}\limits_{i \in V}v_{i}^{\prime}$).
        
    \item The set $V_{-i}$ comprises all nodes in $V$, excluding node $i$ and its descendants $D_{i}^{V}$.
\end{itemize}

\begin{definition}[TRDM]
\label{algo: truthful referral mechanism}
    The Truthful Referral Diffusion Mechanism works as follows
    \begin{enumerate}
        \item Generate a directed graph $G(\theta^{\prime})=(V(\theta^{\prime}), E(\theta^{\prime}))$ from the reported profile $\theta^{\prime}$ (hereafter, denote $(V(\theta^{\prime}), E(\theta^{\prime})$ as $(V, E)$) .

        \item The mechanism selects the highest bidder $w \in \mathop{\arg\max}\limits_{i \in V} v_{i}^{\prime}$ as the winner. (In case of multiple winners, the agent closest to the seller is selected.)
        
        \item The mechanism identifies the shortest path $\Gamma_{w}$. (In case of multiple shortest paths, at the splitting point, select an agent who has a larger valuation in the sub-network without descendants in $\Gamma_{w}$.)

        \item The payment of the winner is defined as $p_{w}(\theta^{\prime})=v_{V_{-w}}^{*}$.

        \item The revenue of the seller is defined as $Rev^{TRDM}(\theta^{\prime})=v_{V_{-r_{s}^{*}}}^{*}$.

        \item The referral reward of agent $i \in \Gamma_{w} \setminus \{s, w\}$ is defined as $p_{i}(\theta^{\prime})=v_{V_{-i}}^{*}-v^{*}_{V_{-r_{i}^{*} }}$; otherwise, $p_{i}(\theta^{\prime})=0$ for $i \notin \Gamma_{w}$.
    \end{enumerate}
\end{definition}

Note that in steps $2$ and $3$, random tie-breaking is applied when multiple options exist.

The process of TRDM is straightforward. It works similarly to salesmen (brokers), where agents purchase the item from the seller and resell it to their children in order to earn a profit margin. However, in TRDM, all the transactions are centralized, and only the winner $w$ makes a payment, while others are rewarded by the seller. In addition, the reward is the highest valuation difference between the network without him and the network without his child, who is also in $\Gamma_{w}$.

\begin{figure}[htb!]
    \centering
    \subfloat[Random network]{%
        \centering
        \includegraphics[width=.35\textwidth]{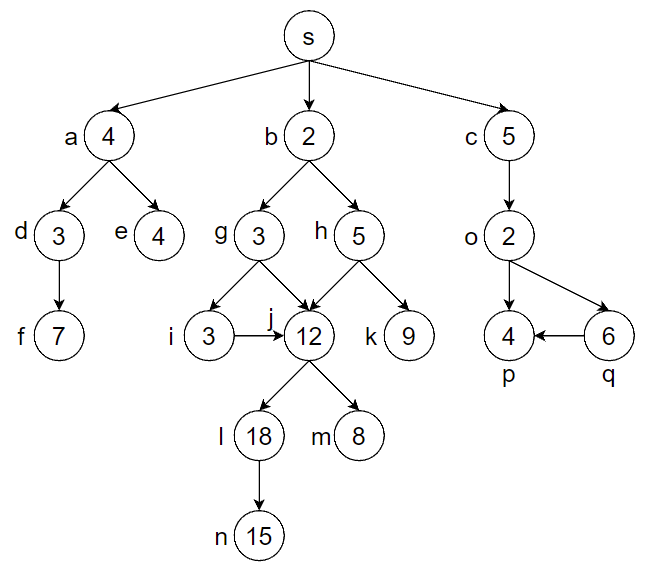}%
        \label{subfig: social network}%
    }\\
    \subfloat[Process of TRDM]{%
        \centering
        \includegraphics[width=.35\textwidth]{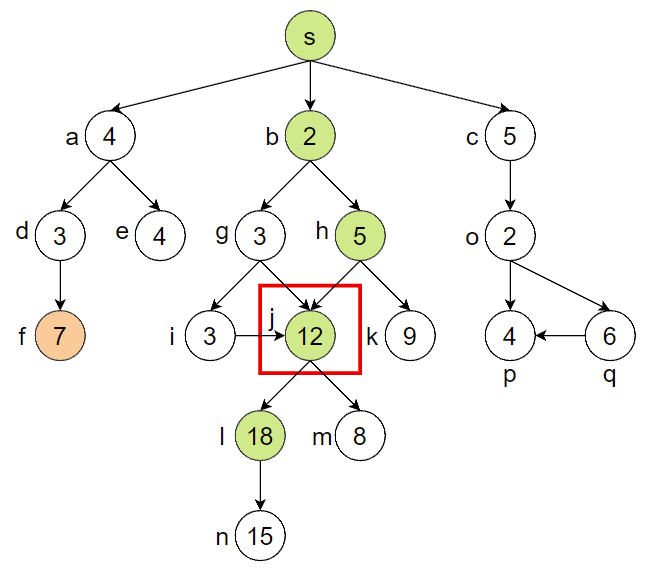}%
        \label{subfig: example}%
    } 
    \caption{A random social network and the corresponding TRDM process. The number in the circle represents the bid price.}
    \label{fig: social network}
\end{figure}

We demonstrate a running example of TRDM by using the network in Figure \ref{subfig: social network}. Agent $l$ has the highest bid in the network and denotes him as a winner $w$. 

Then, we identify the shortest path from the seller to the winner, $\Gamma_{w}$. However, there are two shortest paths such as $\Gamma_{w}=\{s, b, g, j, l(w) \}$ and $\Gamma_{w}=\{s, b, g, h, l(w) \}$. Considering the sub-networks rooted at agent $g$ and $h$, respectively, removing agents who are in $\Gamma_{w}$, agent $k$'s valuation is higher than that of agent $i$. Since agent $k$ is invited by agent $h$, we select the path which includes agent $h$ and plot it in green. (Recall step $3$.)

The red box represents the maximum bid in the network without the participation of the winner $w$ (agent $l$) and his descendants ($n$). It is also the payment of the winner $w$. The orange circle is the maximum bid in the network without the participation of agent $b$ and his descendants ($g$, $h$, $i$, $j$, $k$, $l$, $m$, and $n$). It is also the revenue of the seller.

The payment (reward) of agent $j$ is $9-12=-3$, where $9$ is the valuation of agent $k$, which is the maximum bid in the network without agent $j$ and his descendants, $12$ is the bid of agent $j$,  which is the maximum bid in the network without agent $l$ and his descendants. Similarly for other agents $\{b, h, j \}$, the rewards are $\{0, 2, 3 \}$.

\subsection{Properties of TRDM}
\label{sec: property}

In this section, we discuss the properties that the Truthful Referral Diffusion Mechanism achieves.

\begin{lemma}
\label{lemma: IC and IR}
    TRDM is individually rational and partially incentive compatible, which agents refer truthfully (e.g., $\theta_{i}^{\prime}=(v_{i}^{\prime}, r_{i})$).
\end{lemma}

\begin{proof}
    (\textbf{IR}) Assume agent $i$ reports truthfully $\theta_{i}$. We first consider he is not allocated the item. If agent $i$ is not in the shortest path of winner $w$ ($i \notin \Gamma_{w}$), according to the payment policy, his payment is $0$. If he is in $\Gamma_{w}$, his utility can be $u_{i}(\theta)=x_{i}(\theta)v_{i}-p_{i}(\theta)=-p_{i}(\theta)$. Since he is not the winner, according to the rule of TRDM, only the winner pays the seller. Hence, $p_{i}(\theta) \leq 0$ and $u_{i}(\theta) \geq 0$.
    
    Now we consider the case when agent $i$ wins the item. Based on the payment rule, the winner pays the price bid by the second-highest bidder, which is never higher than his valuation. (i.e., $p_{w}(\theta)=v_{V_{-w}}^{*} \leq v_{w}$) Therefore, his utility $u_{w}(\theta)=v_{w}-p_{w} \geq 0$, and TRDM is individually rational.
    
    (\textbf{partially IC}) \textbf{Case 1:} We first consider agent $i$ has the largest valuation in the network started from him. According to the payment policy, if he is not the winner, he is also not the winner's ancestor. As a result, he is not rewarded, and misreporting will not change his utility.
    
    If he is the winner, the payment only depends on the network without him and his descendants. Hence, misreporting $v_{i}^{\prime}$ and $r_{i}^{\prime}$ will not change the allocation and his utility. 
    
    Moreover, if $v_{i}^{\prime} < v_{i}$, he may lose the item, which leads to zero utility. If $v_{i}^{\prime} > v_{i}$, he may purchase the item at a price that is higher than $v_{i}$, which leads to a negative utility. Therefore, misreporting $\theta_{i}^{\prime}=(v_{i}^{\prime}, r_{i}^{\prime})$ would not be beneficial for agent $i$.
    
    \textbf{Case 2:} Now, we consider the case where agent $i$ does not have the largest valuation in the network started from him. If he is not the ancestor of the winner ($i \notin \Gamma_{w}$), whatever he reports $v_{i}^{\prime}$ and $r_{i}^{\prime}$, his utility is always $0$. and hence misreporting will not improve his utility.

    If he is the ancestor of the winner ($i \in \Gamma_{w}$), let $j$ be the child of $i$ (e.g., $j \in r_{i}$) and $j \in \Gamma_{w}$. Based on the rule of TRDM, the reward of agent $i$ is related to the network without him ($v_{V_{-i}}^{*}$) and the sub-network without agent $j$ ($v_{V_{-j}}^{*}$). If the highest valuation on the sub-network without agent $j$ is exactly the valuation of agent $i$ (e.g., $v_{V_{-j}}^{*}=v_{i}$), then agent $i$ can misreport $v_{i}^{\prime}=[v_{i}, v_{V_{-w}}^{*})$ to improve his utility. More details on the bidding strategy of these agents are discussed later.

    However, agent $i$ cannot improve his utilities by misreporting $r_{i}^{\prime}$. 

    \textbf{Case 2a}: Consider the case that winner $w$ such that $w \in r_{i}$ and $w \notin r_{i}^{\prime}$, if he misreports $r_{i}^{\prime}$, then there exists a new winner $w^{*}$ and he is not in the new shortest path $\Gamma_{w^{*}}$ (e.g., $i \notin \Gamma_{w^{*}}$). As a result, agent $i$ is not rewarded by the seller. 

    \textbf{Case 2b}: If agent $i$ becomes the winner or the ancestor of the new winner, his new utility is not higher than that when he refers truthfully, as the valuation of the new winner is lower than the previous winner.
    
    In conclusion, TRDM is partially IC, with all agents referring truthfully (e.g., $\theta_{i}^{\prime}=(v_{i}^{\prime}, r_{i})$ for all $i \in \Gamma_{w} \setminus \{s, w\}$).
\end{proof}

Lemma \ref{lemma: IC and IR} reveals that it is always optimal for agents to participate in the auction and refer their children truthfully under TRDM. In addition, as TRDM is a truthful referral mechanism, according to Corollary \ref{coro: sybil-collusion-truthful}, agents cannot improve their utility by applying Sybil attacks or forming a collusion group with others under TRDM.

\begin{corollary}
\label{coro: TRDM sybil and collusion}
    TRDM is Sybil-proof and collusion-proof.
\end{corollary}

Due to the fact that TRDM is a non-truthful bid mechanism, we investigate its equilibrium. It is important to note that TRDM is partially IC, with all agents referring truthfully. This simplifies the analysis of equilibrium.

\begin{proposition}
\label{prop: equlibrium}
    Under TRDM with agents $i \in N$, given a real non-negative parameter $\epsilon$, if all agents know the highest valuation among their descendants, there exists an $\epsilon$-approximate Nash equilibrium in which   
    \begin{itemize}
        \item $r_{i}^{\prime}=r_{i}$ for $i \in N$ (all agents refer truthfully).
        \item the bid function $b_{i}(v_{i}, v_{D_{i}^{V}}^{*})$ for $i \in N$ is given by
        \begin{equation}
        \label{eq: bid function}
            v_{i}^{\prime} = b_{i}(v_{i}, v_{D_{i}^{V}}^{*})  = max(v_{i}, v_{D_{i}^{V}}^{*}-\delta),
        \end{equation}
        where $D_{i}^{V}$ is a set of agent $i$'s descendants in $V$, $v_{D_{i}^{V}}^{*}$ is the highest valuation in the set $D_{i}^{V}$, and an arbitrary small positive $\delta$.
    \end{itemize}
\end{proposition}

\begin{proof}
    As a part of the proof of Lemma \ref{lemma: IC and IR}, all agents refer truthfully ($r_{i}^{\prime}=r_{i}$). 
    
    Note that based on the rule of TRDM, the winner has the highest valuation in the network. Hence, $max(v_{w}, v_{D_{w}^{V}}^{*}) = v_{w} \geq v_{D_{w}^{V}}^{*}$ always holds.

    \textbf{We begin with agent $w$ with the highest valuation in the sub-network rooted at him.} 
    
    \textbf{(Case 1)} If he agent $w$ not the winner, then he is also not in the shortest path to the winner $\Gamma_{w}$. For such an agent, whatever he reports, he is never rewarded and his utility is always $0$.
    
    \textbf{(Case 2)} If agent $w$ is the winner, according to the rule of TRDM, the winner pays the highest value in the network without the winner and his descendants ($v_{-w}^{*}$). 
    
    If he bids honestly ($v_{w}^{\prime}=v_{w}$), he wins and pays $v_{-w}^{*}$. If he overbids ($v_{w}^{\prime}>v_{w}$), he still wins and pays the same amount $v_{-w}^{*}$. 
    
    If he underbids ($v_{w}^{\prime}<v_{w}$), there exists three results,
    \begin{enumerate}
        \item he still wins and pays $v_{-w}^{*}$,
        \item he loses and the new winner $w^{\prime}$ is his child ($w^{\prime} \in r_{w}$),
        \item he loses and not in the new shortest path of the new winner $\Gamma_{w^{\prime}}$.
    \end{enumerate}
    
    For Result (2), he receives a reward from the seller, and new utility is $v_{-w^{\prime}}^{*}-v_{-w}^{*}$. Moreover, as $v_{w}^{\prime} < v_{-w^{\prime}}^{*} \leq v_{w}$, underbidding leads to a lower utility than reporting truthfully, where the utility is $v_{w}-v_{-w}^{*}$.
    
    For Result (3), the problem goes back to Case 1, and his utility also decreases.
    
    Therefore, for agents with the highest valuation in the sub-networks rooted at them, the optimal strategy is to report their valuation truthfully.

    \textbf{Now, we consider all other agents who do not have the largest valuation in the sub-network rooted at them.} Recall that agents are only rewarded if they are in the shortest path to the winner $\Gamma_{w}$. Moreover, agents know the highest valuation among their descendants set $v_{D_{i}^{V}}^{*}$.

    Let $v_{i} < v_{D_{i}^{V}}^{*}$ such that one of agent $i$'s descendants (agent $w$) has the highest valuation in the network rooted at agent $i$. Then, agent $i$ knows he cannot be the winner of the item, and agent $w$ may become the winner. (Recall the rule of TRDM, the agent is rewarded based on the highest value in the sub-network without his child in $\Gamma_{w}$.)
    
    \textbf{(Case 1)} If agent $w$ becomes the winner, agent $i$ needs to bid the second-highest value in the network to maximize his reward. Moreover, the closer his bid is to the winner's bid, the higher his utility will be.
    
    Thus, such agents bid $v_{i}^{\prime}=v_{D_{i}^{V}}^{*}-\delta$, where $\delta$ is an arbitrary small positive number. Consequently, their utility is $v_{D_{i}^{V}}^{*}-\delta-v_{-i}^{*}$. Since $\delta$ is a small positive number, there always exists a non-negative $0<\epsilon<\delta$ such that the utility of bid $v_{D_{i}^{V}}^{*}-\epsilon$ is greater than $v_{D_{i}^{V}}^{*}-\delta$. ($\epsilon$-approximate Nash equilibrium)
    
    \textbf{(Case 2)} If agent $w$ is not the winner, then agent $i$ is not in the path $\Gamma_{w}$. As a result, whatever agent $i$ reports, agent $i$ is not rewarded.

\end{proof}

We refer readers to page $45$ of the book \textit{Algorithmic Game Theory} \cite{nisan2007algorithmic} for details on $\epsilon$-approximate Nash equilibrium.

Proposition \ref{prop: equlibrium} states that agents bid the highest valuation in the sub-network rooted at them. However, overbidding does not affect the seller's revenue. The seller's revenue is determined by the highest valuation in the sub-network without $\Gamma_{w}$. Even if some agents in the sub-network overbid, they will never bid higher than the highest valuation in the corresponding network. As a result, the seller's revenue is not affected by misreporting.

Furthermore, according to Proposition \ref{prop: equlibrium}, we can easily derive that agent who is both the child of the seller and the ancestor of the winner ($i=r_{s}^{*}$) is rewarded the most.

\begin{corollary}
\label{coro: max reward}
    Under TRDM, if all agents bid optimally, the agent $i \in r_{s}^{*}$ has the highest utility, where $r_{s}^{*}$ is the children set of the seller $s$ and also located in the winning path $\Gamma_{w}$.
\end{corollary}

In what follows, we analyze the allocation and revenue outcomes of TRDM. We demonstrate that TRDM allocates the item to the agent with the highest valuation and guarantees the seller a non-negative revenue.

\begin{lemma}
\label{lemma: efficient and budget feasible}
    TRDM is efficient and budget balanced.
\end{lemma}

\begin{proof}
    (\textbf{Efficient}) According to the allocation rule of TRDM, the item is allocated to the bidder with the highest valuation. Hence, TRDM is efficient. (see Definition \ref{def: efficient})

    (\textbf{Budget balance}) As mentioned in Definition \ref{algo: truthful referral mechanism}, the seller's revenue under TRDM is fixed and $Rev=v_{V_{-r_{s}^{*}}}^{*} \geq 0$. Hence, TRDM is budget-balanced. (see Definition \ref{def: budget-balance})
\end{proof}

TRDM satisfies all the desirable properties outlined in Section \ref{sec: preliminaries and model}, except IC. However, partially IC (non-truthful bid) is not a significant concern in TRDM, as the revenue is fixed and not influenced by overbidding.

\section{Comparison}
\label{sec: comparison}

It is important to note that the network and valuations used in this analysis are randomly generated for the purpose of comparison. Our findings and conclusions hold for any network, regardless of its specific structure or valuations. 

To provide an intuitive comparison of TRDM with other existing diffusion mechanisms, we present a running example in Figure \ref{fig: compare example} and provide numerical results in Table \ref{tab: comparison}. For consistency, we assume all agents report truthfully in all mechanisms.

In this section, we compare the performance of existing diffusion mechanisms such as VCG, IDM \cite{li2017mechanism}, CDM \cite{li2019diffusion}, FDM \cite{zhang2020incentivize}, NRM \cite{zhang2020redistribution}, and TRDM.

\begin{figure}[h!]
    \centering
    \includegraphics[width=0.35\textwidth]{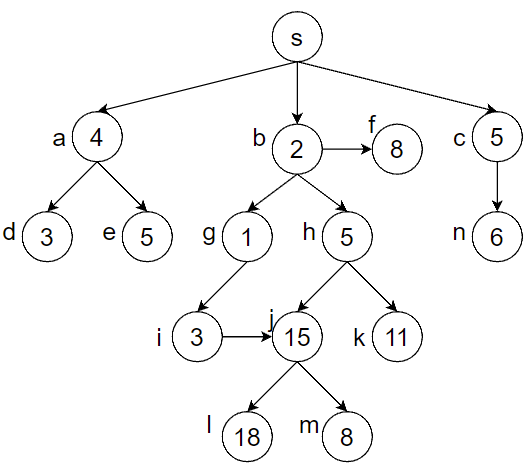}
    \caption{Network for comparison}
    \label{fig: compare example}
\end{figure}

\begin{table}[h!]
\centering
\resizebox{0.5\textwidth}{!}{
\begin{tabular}{|c|c|c|c|c|}
\hline
\multirow{2}{*}{Mechanism} & \multirow{2}{*}{Winner} & \multirow{2}{*}{Rewarded agents} & \multirow{2}{*}{Social Welfare} & \multirow{2}{*}{Revenue} \\
 &  &  &  &  \\ \hline
\multirow{2}{*}{\textbf{VCG}} & \multirow{2}{*}{$l(15)$} & \multirow{2}{*}{$b(-12), j(-7)$} & \multirow{2}{*}{18} & \multirow{2}{*}{-4} \\
 &  &  &  &  \\ \hline
\multirow{2}{*}{\textbf{IDM}} & \multirow{2}{*}{$j(11)$} & \multirow{2}{*}{$b(-2), h(-3)$} & \multirow{2}{*}{15} & \multirow{2}{*}{6} \\
 &  &  &  &  \\ \hline
\multirow{2}{*}{\textbf{CDM}} & \multirow{2}{*}{$j(11)$} & \multirow{2}{*}{$b(-2)$} & \multirow{2}{*}{15} & \multirow{2}{*}{9} \\
 &  &  &  &  \\ \hline
\multirow{2}{*}{\textbf{FDM}} & \multirow{2}{*}{$j(\frac{43}{4})$} & \multirow{2}{*}{$b(-2), g(-\frac{3}{4}), h(-\frac{1}{4}), i(-\frac{3}{4})$} & \multirow{2}{*}{15} & \multirow{2}{*}{$7$} \\
 &  &  &  &  \\ \hline
\multirow{4}{*}{\textbf{NRM}} & \multirow{4}{*}{$j(11)$} & \multirow{2}{*}{$a(-\frac{18}{14}), b(-\frac{45}{14}), c(-\frac{10}{14}),$} & \multirow{4}{*}{15} & \multirow{4}{*}{$\frac{39}{14}$} \\
 &  &  &  &  \\
 &  & \multirow{2}{*}{$h(-2), k(-1)$} &  &  \\
 &  &  &  &  \\ \hline
\multirow{2}{*}{\textbf{TRDM}} & \multirow{2}{*}{$l(15)$} & \multirow{2}{*}{$b(-2), h(-3), j(-4)$} & \multirow{2}{*}{18} & \multirow{2}{*}{$6$} \\
 &  &  &  &  \\ \hline
\end{tabular}
}
\caption{Empirical results for different diffusion mechanisms in Figure \ref{fig: compare example}. Numbers in parentheses represent the corresponding payments. Positive for payment and negative for reward.}
\label{tab: comparison}
\end{table}

The main improvement of TRDM is it ensures all agents refer truthfully, which is one of the most important properties of the diffusion mechanism. Both Sybil attacks and collusion groups can negatively impact the seller's revenue. 

For example, under \textbf{IDM}, agent $l$ creates a fake account $l^{*}$ with a bid between $(15, 18]$, and invite himself. Then, his fake account $l^{*}$ becomes the winner and pays $15$ to purchase the item. For agent $l$, he obtains the item, and the corresponding utility becomes from $0$ to $3$ by applying Sybil attacks. 

Under \textbf{CDM}, assume agent $f$ is a fake account created by agent $b$. If agent $f$ bids $11$, his reward increases from $2$ to $11-6=5$.. Furthermore, agent $b$ can also form a collusion group with his descendants ($h$ and $j$) and bid any value higher than $6$. Then, he purchases the item at a price of $6$ and resells it to agent $l$ at a price of $15$. Finally, agents $b$, $h$, and $f$ redistribute the surplus of $15-6=9$. For both methods, the revenue decreases from $9$ to $6$.

Under \textbf{FDM}, the explanation is the same as that of CDM, when agent $b$ creates a fake account agent $f$ with a bid $11$. 

Under \textbf{NRM}, agent $a$ is rewarded $\frac{3*6}{14}=\frac{18}{14}$ if he reports truthfully. ($3$ is the number of his descendants and himself, $14$ is the number of agents in the system, and $6$ is the highest valuation without the network started from agent $b$, who is the ancestor of the winner.) If he creates \textbf{two} other fake ids with any bidding value less than $6$, he is rewarded $\frac{5*6}{16}=\frac{30}{16}>\frac{18}{14}$. (more descendants, higher reward) Hence, it is not Sybil-proof.

Furthermore, the above examples are consistent with Theorem \ref{thm: upper bound} that the seller's revenue is upper bounded; otherwise, agents may apply Sybil bids or collusion groups to improve their utility.

Although IDM is not Sybil-proof, Sybil bids do not affect the seller's revenue, as the revenue is fixed (see \cite{li2017mechanism}). For all other mechanisms, the more Sybil bids are applied, the lower the revenue. Nevertheless, IDM and TRDM are the only two mechanisms in which Sybil bids have no impact (reduction) on the seller's revenue. In addition, TRDM rewards more than IDM,

\begin{proposition}
\label{prop: compare reward}
    The revenue generated by TRDM is equal to that of IDM, while the total reward of TRDM is \textbf{not} lower than that of IDM.
\end{proposition}

Due to space constraints, proofs are provided in Appendix.

The numerical experiment in Table \ref{tab: comparison} shows that TRDM rewards the most compared with other mechanisms, TRDM ($9$) outperforms IDM ($5$), CDM ($2$), FDM ($\frac{15}{4} = 3.75$) and NRM ($\frac{115}{14} \approx 8.2$) in overall reward. As we discussed, under TRDM, some agents overbid, which also increases their total reward without influencing the seller's revenue. 

Moreover, TRDM outperforms any traditional auction without referrals on revenue. This is consistent with the Bulow-Klemperer Theorem, which suggests that increasing competition is often more effective than refining the auction format. It is worth noting that, in network settings, the VCG mechanism is the only traditional auction that incentivizes agents to invite others, even though it may not always be budget balanced.

\begin{proposition}
\label{prop: compare revenue}
    The revenue generated by TRDM is \textbf{not} lower than that of the VCG mechanism, and it reaches the upper bound defined in Theorem \ref{thm: upper bound}.
\end{proposition}

\section{Conclusion}
\label{sec: conclusion}

This paper explores auctions over networks and proposes the Truthful Referral Diffusion Mechanism (TRDM). TRDM guarantees that all agents benefit from participating in the auction and some agents are rewarded for contributions to the improved revenue. Moreover, TRDM ensures that the social welfare is maximized and the revenue generated attains the upper bound of the maximum revenue achievable by a diffusion mechanism, which is also higher compared to auctions without referrals. The key difference of TRDM from prior works is that it emphasizes the importance of truthful referral rather than truthful bidding, which results in an efficient allocation. Although some agents may overbid under TRDM, the revenue remains unaffected, and the total reward improves. As a result, TRDM provides attractive returns and incentives for both sellers and buyers.

Our study offers valuable insights into mechanism design in network settings. However, it would be intriguing to extend the model to consider multiple items.

\clearpage
\bibliography{reference}

\clearpage
\section{Proof of Proposition \ref{prop: compare reward}}
\begin{proof}
    The revenue of IDM is defined as
    \begin{equation*}
        \begin{aligned}
            Rev^{IDM}=v_{V_{-\{\Gamma_{w} \setminus s \}}}^{*}=Rev^{TRDM}.
        \end{aligned}
    \end{equation*}

    Based on the rule of IDM, the item is allocated to the agent who invites the agent with the highest valuation in the social network. Mathematically, $V_{-w_{IDM}} \in V_{-w_{TRDM}}$ implies $v_{V_{-w_{IDM}}}^{*} \leq v_{V_{-w_{TRDM}}}^{*}$.

    Therefore, the total reward of IDM is
    \begin{equation*}
        \begin{aligned}
            Reward^{IDM}&=p_{w_{IDM}}-Rev^{IDM}\\
            &= v_{V_{-w_{IDM}}}^{*} - Rev^{IDM}\\
            & \leq v_{V_{-w_{TRDM}}}^{*} - Rev^{IDM}\\
            &= Reward^{TRDM}
        \end{aligned}
    \end{equation*}
    
   Hence, the overall reward of TRDM is no less than that of IDM.
\end{proof}

\section{Proof of Proposition \ref{prop: compare revenue}}

\begin{proof}
    For the VCG mechanism, we first consider the case when agents do not refer others. For a single-item case, VCG is equivalent to Second Price Auction (SPA). The winner $w_{SPA}$ pays the second highest value. Since $\{r_{s} \setminus w_{SPA}\} \in V_{-\Gamma_{w_{TRDM}}}$, we have $v_{r_{s} \setminus w_{SPA}}^{*} \leq v_{V_{-\Gamma_{w_{TRDM}}}}^{*}$. Intuitively, the second highest value among the social network without agents in the path $\Gamma_{w_{TRDM}}$ is no less than the second highest value among seller's children.
    
    Hence, the revenue of TRDM is no less than the VCG mechanism without referrals.
    
    Secondly, we consider the case when agents refer others. For agents $i \notin \Gamma_{w}$, their payments are always 0. For agents $i \in \Gamma_{w}$, they are rewarded by the seller based on the VCG mechanism. Therefore, the revenue of the VCG mechanism in the social network is

    \begin{equation*}
        \begin{aligned}
            Rev^{VCG} &= \sum_{i \in \Gamma_{w}} p_{i} \\
            &= \sum_{i \in \Gamma_{w} \setminus w} (v_{V_{-i}}^{*}-v_{w})+(v_{V_{-w}}^{*}-0)\\
            &\leq v_{V_{-r_{s}^{*}}}^{*}=Rev^{TRDM} 
        \end{aligned}
    \end{equation*}
    
    Consequently, the revenue of TRDM is no less than that of the VCG mechanism with referrals.

    By step 5 of Definition \ref{algo: truthful referral mechanism}, the revenue of TRDM achieves the upper bound defined in Theorem \ref{thm: upper bound}. 
\end{proof}

\end{document}